\documentclass[a4paper,11pt]{article}
\usepackage[top=25truemm,bottom=25truemm,left=15truemm,right=15truemm]{geometry}
\usepackage{cite}
\usepackage{amsmath,amssymb,amsfonts}
\usepackage{mathtools}
\usepackage{amsthm}
\usepackage{algorithm}
\usepackage{algpseudocode}
\usepackage{textcomp}
\usepackage{bm}
\usepackage[dvipdfmx]{graphicx}
\usepackage{cite}
\usepackage{here}
\usepackage{comment}
\usepackage{latexsym}
\newcommand{\Hline}[1]{\noalign{\hrule height #1}}
\newtheorem{theorem}{Theorem}

\begin{document}
\title{Necessary and Sufficient Conditions for\\Capacity-Achieving Private Information Retrieval\\with Adversarial Servers}
\author{Atsushi Miki, Toshiyasu Matsushima}
\date{}
\maketitle

\begin{abstract}
 Private information retrieval (PIR) is a mechanism for efficiently downloading messages while keeping the index of the desired message secret from the servers.
  PIR schemes have been extended to various scenarios with adversarial servers:
  PIR schemes where some servers are unresponsive or return noisy responses are called robust PIR and Byzantine PIR, respectively;
  PIR schemes where some servers collude to reveal the index are called colluding PIR.
  The information-theoretic upper bound on the download efficiency of these PIR schemes has been proved in previous studies.
  However, systematic ways to construct PIR schemes that achieve the upper bound are not known.
  In order to construct a capacity-achieving PIR schemes systematically, it is necessary to clarify the conditions that the queries should satisfy.
  This paper proves the necessary and sufficient conditions for capacity-achieving PIR schemes. 
\end{abstract}

\section{Introduction}
\label{sec:introduction}
The growing scale of database systems for cloud computing and big data has drawn attention to the need for privacy protection.
Private information retrieval (PIR) is an important mathematical model for protecting privacy.
PIR is a mechanism for retrieving a message from database servers 
through efficient communication while keeping the index of the desired message secret from the servers.
The privacy considered in this paper is information-theoretic, not computational, security.
The simplest technique to preserve privacy is to download the entire library from a single server, but the download efficiency is low.
It is known that utilizing multiple servers with identical functionality and encoding queries can reduce the communication cost.
The communication cost is defined as the entropy of sending a query (hereinafter, upload cost) and receiving a response (hereinafter, download cost).

PIR performance is mainly evaluated by download efficiency (hereinafter, rate).
The rate is defined as the ratio of the desired message length to the download cost.
The information-theoretic upper bound of the rate (hereinafter, capacity) was derived in \cite{7889028}.
PIR capacity is an important metric from the perspective of information theory.
Therefore, previous studies have focused on constructing methods that achieve capacity.
In other words, PIR has three key properties: privacy to protect the user's choice, correctness to ensure the desired message can be decoded, and capacity achievability.

Meanwhile, PIR models have been extended to various scenarios.
Representative extensions include 
robust PIR \cite{8119895,8598994}, where responses from some servers are not received; 
Byzantine PIR \cite{8457293,8613332,8262858}, where some servers return incorrect responses; 
colluding PIR \cite{8119895,8262858,9865995}, where some servers share information.
Collectively, these scenarios are referred to as adversarial PIR.
These extended scenarios are highly compatible with linear PIR, as they utilize the properties of maximum distance separable (MDS) codes, a type of linear code.
Linear PIR refers to schemes in which the servers' responses are constructed as matrix products of the messages and queries.

Previous studies on extended scenarios classified as linear PIR include Sun's method \cite{8119895} and Wang's method \cite{8262858}.
Sun's method is known to be extendable to robust PIR, Byzantine PIR, and colluding PIR, and furthermore achieves capacity in each of these scenarios.
Wang's method can also be extended to similar scenarios, but does not achieve capacity.

Conventionally, there have been no systematic construction methods for capacity-achieving PIR schemes, and even when a PIR method is devised, it is often difficult to verify whether it satisfies the PIR properties: privacy, correctness, and capacity.
This study aims to clarify the necessary and sufficient conditions that queries must satisfy to achieve capacity in linear PIR.
By expressing the capacity conditions in terms of the properties of the query matrix, we can verify whether a devised PIR scheme satisfies the PIR properties.
Furthermore, this provides guidelines for constructing capacity-achieving PIR schemes.

In this paper, we derive the necessary and sufficient conditions 
that query matrices must satisfy to achieve capacity for each scenario of adversarial PIR schemes.
Furthermore, we apply these conditions to conventional methods to verify whether the PIR properties are satisfied.

The remainder of this paper is organized as follows: 
Section \ref{sec:notation} defines the notations used in this paper;
Section \ref{sec:problemstatement} describes the problem setup of adversarial PIR schemes;
Section \ref{sec:representative} introduces representative schemes that are applicable to extended scenarios;
Section \ref{sec:conditions} derives the necessary and sufficient conditions for PIR properties;
Section \ref{sec:confirmation} applies the derived conditions to conventional methods to verify whether PIR properties are satisfied.

\section{Notation}
\label{sec:notation}
In this section, we define the notations used in this paper.

For $i,j \in \mathbb{Z}$, $[i:j] \coloneqq \{i, i+1, \dots , j\}$; if $i > j$, $[i:j]$ is the empty set $\emptyset$. 
For $n,z \in \mathbb{N}$, a vector $\mathcal{V}=(v[1],\dots,v[z])$ and a set $\mathbf{s}=\{s_1,\dots,s_n\} \subset [1:z]$, $\mathcal{V}[{\mathbf{s}}]$ denotes $(v[s_1],\dots,v[s_n])$.
For a $r \times c$ matrix $\mathcal{M}$, its row vectors are denoted by $\mathcal{M}_1,\dots,\mathcal{M}_r$, 
and for sets $\mathbf{t}=\{t_1,\dots,t_i\} \subset [1:r], \mathbf{u}=\{u_1,\dots,u_j\} \subset [1:c]$,
$\mathcal{M}_{\mathbf{t}}$ denotes 
$\begin{pmatrix}
  \mathcal{M}_{t_1}\\
    \vdots \\
  \mathcal{M}_{t_i} 
\end{pmatrix}$, 
and 
$\mathcal{M}_{\mathbf{t}}[\mathbf{u}]$ denotes 
$\begin{pmatrix}
  \mathcal{M}_{t_1}[\mathbf{u}]\\
    \vdots \\
  \mathcal{M}_{t_i}[\mathbf{u}] 
\end{pmatrix}$.
For an integer set $\mathcal{K}=\{i_1,\dots,i_n\}$, the vector $(A_{i_1},\dots , A_{i_n})^\top$ is denoted by $A_{\mathcal{K}}$;
for $n_1, n_2 \in \mathbb{Z}$, $A_{[n_1:n_2]}$ is denoted by $A_{n_1:n_2}$.
For a finite field that has order $p$, we use the notation $\mathbb{F}_p$.

\section{Problem Statement}\label{sec:problemstatement}
We consider $S$ servers such that each server stores $M$ messages $\mathcal{W}_{1:M}$ 
which are random variables satisfying the following conditions: 
\begin{align}
 &M \geq 2,\ S\geq 2,
  \\& H(\mathcal{W}_{1:M}) = H(\mathcal{W}_1) + \dots + H(\mathcal{W}_M),
  \\ & H(\mathcal{W}_1) = \dots = H(\mathcal{W}_M) = L_w.
\end{align}
The flow of the process is shown in Fig. \ref{PIR_flow}.
Let $W_{1:M}$ be the realizations of $\mathcal{W}_{1:M}$.
A user wishes to retrieve a message $W_m\ (m\in [1:M])$
while keeping the selected index $m$ secret from the servers.

For all $i \in [1:S]$, the user sends the query $Q_i^{(m)}$ to the $i$-th server.
Let $\mathcal{Q}_i^{(m)}$ be the random variable corresponding to $Q_i^{(m)}$, and the realization set is defined as $\mathbf{Q}_i = \{Q_i^{(m)} | m \in [1:M]\}$. 
The $i$-th server generates a response $\mathcal{X}_i^{(m)}$ 
using $\mathcal{Q}_i^{(m)}$ and $\mathcal{W}_{1:M}$. 
$\mathcal{X}_i^{(m)}$ is also a random variable because $\mathcal{Q}_i^{(m)}$ and $\mathcal{W}_{1:M}$ are random variables.
Let $X_{i}^{(m)}$ be the realization of $\mathcal{X}_{i}^{(m)}$, and let its realization set be $\mathbf{X}_i = \{X_i^{(m)} | m \in [1:M]\}$.
The user decodes $W_m$ using the responses $X_{1:S}^{(m)}$.
Therefore, the following equation should be satisfied for any $m \in [1:M]$.
\begin{eqnarray}
 [\mathrm{Correctness}] \quad H(\mathcal{W}_m|\mathcal{X}_{1:S}^{(m)},\mathcal{Q}_{1:S}^{(m)})=0. \label{correctness}
\end{eqnarray}
\begin{figure}[tb]
  \begin{center}
  \includegraphics[width=9cm]{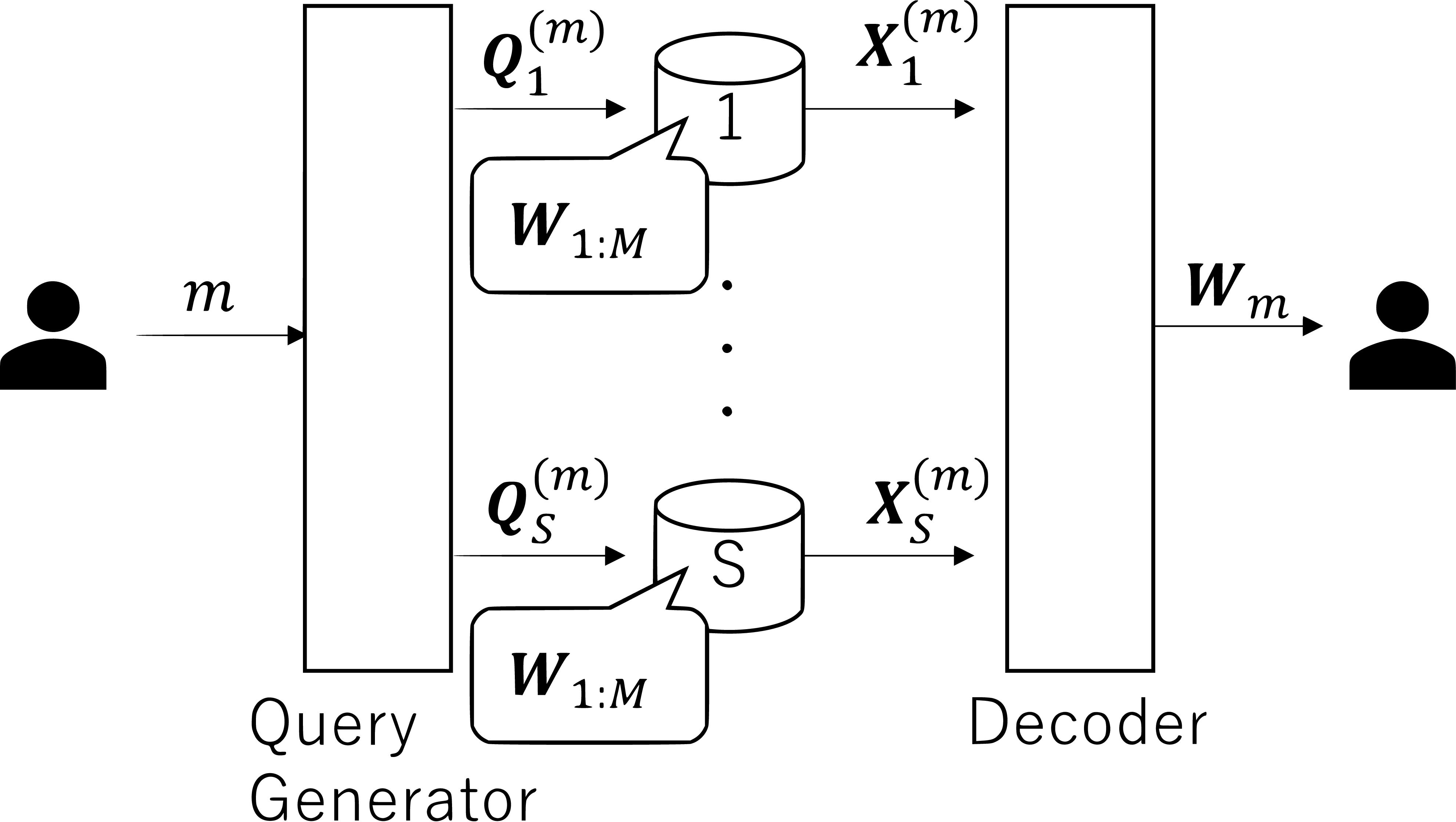}
  \end{center}
  \caption{PIR flow}
  \label{PIR_flow}
\end{figure} 

In colluding PIR, assume $T\ (T\geq 2)$ of the $S$ servers collude to share their information.
Let $\theta$ be a random variable of $m$, which is chosen independently of $\mathcal{W}_{1:M}$ according to a uniform distribution over $[1:M]$, 
and let $\mathcal{T}$ be an index set of $T$ colluding servers, 
then the following equation should be satisfied to ensure privacy for any $\mathcal{T} \subset [1:S]$ satisfying $|\mathcal{T}| = T$
, which means that $\theta$ is not obtained even with all information known by any $T$ colluding servers:
\begin{align}
  [\mathrm{Colluding\ Privacy}] \quad I(\theta;\mathcal{Q}_\mathcal{T}^{(\theta)},\mathcal{W}_{1:M})=0. \label{privacy}
 \end{align}
Let the length of the response from $i$-th server be $L(X_i^{(m)})$, the download rate is defined as follows:
\begin{eqnarray}
 \mathcal{R} = \frac{H(\mathcal{W}_m)}{\sum_{i=1}^{S} \sum_{X_i^{(m)} \in \mathbf{X}_i}\mathrm{Pr}\{\mathcal{X}_i^{(m)}=X_i^{(m)}\}L(X_i^{(m)})}. \label{rate} 
\end{eqnarray}
The capacity of colluding PIR was derived in \cite{8119895} as follows:
 \begin{eqnarray}
  [\mathrm{Colluding\ Capacity}] \quad \mathcal{R} \leq \frac{1-\frac{T}{S}}{1-(\frac{T}{S})^M}. \label{colludingcapacity}
\end{eqnarray} 

For robust PIR, the user uses $S+U$ servers that includes $U\ (U > 0)$ unresponsive servers;
For Byzantine PIR, the user uses $S+2B$ servers that includes $B\ (B > 0)$ noisy servers that return wrong responses; 
these schemes are handled by erasure correction and error correction, respectively.
The privacy properties of these schemes are defined as follows:
\begin{align}
  [\mathrm{Robust\ Privacy}]\quad \forall j\in [1:S+U], I(\theta;\mathcal{Q}_{j}^{(\theta)},\mathcal{W}_{1:M})=0. \label{robustprivacy}
 \end{align}
 \begin{align}
  [\mathrm{Byzantine\ Privacy}]\quad \forall j\in [1:S+2B], I(\theta;\mathcal{Q}_{j}^{(\theta)},\mathcal{W}_{1:M})=0.\label{byzantineprivacy}
 \end{align}
The Capacity achievabilities of these schemes are derived as follows \cite{8119895,8457293}:
 \begin{align}
  &[\mathrm{Robust\ Capacity}] \quad \mathcal{R} \leq \frac{1-\frac{1}{S}}{1-(\frac{1}{S})^M}.
  \\&[\mathrm{Byzantine\ Capacity}] \quad \mathcal{R} \leq \frac{S}{S+2B}\frac{1-\frac{1}{S}}{1-(\frac{1}{S})^M}.\label{byzantinecapacity}
\end{align}

\section{Related Work} \label{sec:representative}
In this section, we introduce representative methods for each adversarial PIR scheme.
\subsection{Sun's Colluding PIR \cite{8119895}} \label{sun'scolludingPIR}
In this method, for all $k\in [1:M]$ 
and for a positive integer $n$, the message is chosen as $W_k \in \mathbb{F}_{S^M}^{nS^M}$\footnote{In \cite{8119895}, the finite field is denoted by $\mathbb{F}_q$, and the order $q$ is defined as $q \geq S^M$. In this paper, we set $q=S^M$ to reduce $\mathcal{U}$.}. 
Therefore, $L_w$ is a multiple of $S^M\log S^M$, and $S$ must be a prime power. 
For all $i\in [1:S]$, the query is chosen as $Q_i^{(m)} \in \mathbb{F}_{S^M}^{\frac{S^M-T}{S-T} \times MS^M}$.
We present an example for $n=1$.

In the query generation phase, for all $i\in [1:S],\ j\in [1:M],\ k\in[1:M]$, a user generates $\mathcal{M}_i^{(m)}$ as follows:
\begin{align}
  &\mathcal{M}_i^{(m)}=(\mathcal{M}_{i,1}^{(m)\top},\ \dots,\ \mathcal{M}_{i,M}^{(m)\top})^\top, \notag\\
  &\mathcal{M}_{i,j}^{(m)} = (\mathcal{M}_{i,j}^{(m)}[1],\ \dots,\ \mathcal{M}_{i,j}^{(m)}[M]),\notag \\ 
  &\mathcal{M}_{i,j}^{(m)}[k]= (v_{i,j,k}^{(m)}[1]^{\top},\ \dots,\ v_{i,j,k}^{(m)}[\begin{pmatrix}M\\j\end{pmatrix}T^{M-j}(S-T)^{j-1}]^{\top})^\top.\label{sun}
\end{align}
For all $l \in \left[1:\begin{pmatrix}M\\j\end{pmatrix}T^{M-j}(S-T)^{j-1}\right]$, $v_{i,j,k}^{(m)}[l]$ is an $S^M$ length vector.
Furthermore, for all $k\in[1:M]$, the user generates 
$\mathbf{s}_k$, which are $S^M \times S^M$ full-rank random matrices.
The user generates $\mathbf{e}_m=\mathbf{s}_m$,
and for all $k\in[1:M]\setminus \{m\}$, the user generates $\mathbf{e}_k$ as follows:
\begin{align}
  \mathbf{e}_k = \mathcal{G}\left(S^{M},TS^{M-1}\right)\mathbf{s}_k[1:TS^{M-1}]^\top,
\end{align}
where $\mathcal{G}(n,k)$ is an $n \times k$ encoding matrix formed by combining generator matrices of MDS codes on a sufficient large finite field $\mathbb{F}_q$.
For all $k\in[1:M]\setminus \{m\}$, $\mathbf{s}_k[1:TS^{M-1}]^\top$ is the first $TS^{M-1}$ rows of $\mathbf{s}_k$. 
Therefore, $\mathbf{e}_k$ is an $S^M \times S^M$ matrix.
The queries are generated as follows:
\begin{align}
  \begin{pmatrix}Q_1^{(m)}\\Q_2^{(m)}\\\vdots \\ Q_S^{(m)}\end{pmatrix} 
  = \begin{pmatrix}\mathcal{M}_1^{(m)}\\\mathcal{M}_2^{(m)}\\\vdots \\ \mathcal{M}_S^{(m)}\end{pmatrix} 
  \begin{pmatrix} 
    \mathbf{e}_1 \quad \quad         &              &                             \\
        \quad \ddots \quad           &              &     \mathbf{O}              \\
        \quad \quad \mathbf{e}_{m-1} &              &                             \\
                                     &\mathbf{e}_m  &                             \\
                                     &              &\mathbf{e}_{m+1} \quad \quad \\
           \mathbf{O}                &              &           \quad \ddots \quad\\ 
                                     &              &           \quad \quad  \mathbf{e}_M 
  \end{pmatrix}. \label{sunsquery}
\end{align}
$\mathcal{M}_1^{(m)},\ \dots,\ \mathcal{M}_S^{(m)}$ 
in (\ref{sun}) are constructed such that there exists a matrix $D_m$ satisfying the following equation:
\begin{align}
  D_m \begin{pmatrix}Q_1^{(m)}\\\vdots \\ Q_S^{(m)}\end{pmatrix} = 
  \begin{pmatrix}
    \mathbf{O}_{(m-1)S^M \times S^M} \\
    \mathbf{E}_{S^M} \\
    \mathbf{O}_{(M-m)S^M \times S^M}
   \end{pmatrix}^\top,\label{sundecode}
\end{align}
where $\mathbf{O}_{l_1 \times l_2}$ denotes the $l_1 \times l_2$ zero matrix and 
$\mathbf{E}_{l}$ denotes the $l \times l$ identity matrix.

In the response generation phase, for all $k\in [1:M]$, each server divides the messages into $S^M$ equal-length parts as follows:
\begin{equation}
  W_k=
(w_{k,1} ,\ \dots ,\ w_{k,S^M})^\top.
\end{equation}
Next, for all $i\in [1:S],\ j\in [1:M],\ k\in[1:M]$, the $i$-th server generates a response as follows:
\begin{align}
  &X_i^{(m)}=(X_{i,1}^{(m)\top},\ \dots,\ X_{i,M}^{(m)\top})^\top,\\
  &X_{i,j}^{(m)}=\bigl(\sum_{k=1}^{M} v_{i,j,k}^{(m)}[1]\mathbf{e}_k\begin{pmatrix}w_{k,1}\\\vdots\\w_{k,S^M}\end{pmatrix},&\dots,\ \sum_{k=1}^M v_{i,j,k}^{(m)}[\begin{pmatrix}M\\j\end{pmatrix}T^{M-j}(S-T)^{j-1}]\mathbf{e}_k\begin{pmatrix}w_{k,1}\\\vdots\\w_{k,S^M}\end{pmatrix}\bigr)^\top.
\end{align}

In the decoding phase, the user decodes the desired message $W_m$ as follows:
\begin{align}
  D_m \begin{pmatrix}X_1^{(m)}\\\vdots \\ X_S^{(m)}\end{pmatrix} = \begin{pmatrix}w_{m,1}\\\vdots\\w_{m,S^M}\end{pmatrix} = W_m.
\end{align}
Because $D_m$ satisfies (\ref{sundecode}), correctness is guaranteed.

\subsection{Sun's Robust PIR \cite{8119895}}
In the scheme described in \ref{sun'scolludingPIR}, 
the order of the finite field is changed to a prime power $q \geq (S+U)S^{M-1}$,
and the user generates $\mathbf{e}_m=\mathcal{C} \mathbf{s}_m$, where $\mathcal{C}$ is an $((S+U)S^{M-1} , S^M)$ generator matrix of an MDS code;
further, for all $k\in[1:M]\setminus \{m\}$, the user generates $\mathbf{e}_k$ as follows:
\begin{align}
  \mathbf{e}_k = \mathcal{G}\left((S+U)S^{M-1},TS^{M-1}\right)\mathbf{s}_k[1:TS^{M-1}]^\top.
\end{align}
In this case, the user sends queries to $S+U$ servers, and the responses from any $S$ servers are used for decoding.
The user can decode $W_m$ by erasure correction, because the responses are MDS-encoded.
By setting $T=1$, this method is equivalent to non-colluding robust PIR.

\subsection{Banawan's Byzantine PIR \cite{8457293}}
This method has a same mechanism as Sun's robust PIR\footnote{Hereinafter, we regard this method as equivalent to Sun's robust PIR.}.
In the scheme described in \ref{sun'scolludingPIR}, 
the order of the finite field is changed to a prime power $q \geq (S+2B)S^{M-1}$,
and the user generates $\mathbf{e}_m=\mathcal{C} \mathbf{s}_m$, where $\mathcal{C}$ is an $((S+2B)S^{M-1} , S^M)$ generator matrix of an MDS code;
further, for all $k\in[1:M]\setminus \{m\}$, the user generates $\mathbf{e}_k$ as follows:
\begin{align}
  \mathbf{e}_k = \mathcal{G}\left((S+2B)S^{M-1},TS^{M-1}\right)\mathbf{s}_k[1:TS^{M-1}]^\top.
\end{align}
In this case, the user sends queries to $S+2B$ servers.
The user can decode $W_m$ by error correction, because the responses are MDS-encoded.
By setting $T=1$, this method is equivalent to non-colluding Byzantine PIR.

\subsection{Wang's Colluding PIR \cite{8262858}} \label{wan'scolludingPIR}
In this method\footnote{This method can be applied to symmetric PIR, and eavesdropper scenario. However, in this paper, we set the common randomness to zero and assume no eavesdroppers.}, for all $k\in [1:M]$, and for a positive integer $n$, the message is chosen as $W_k \in \mathbb{F}_S^{n(S-T)}$. 
Therefore, $L_w$ is a multiple of $(S-T)\log S$, and $S$ must be a prime power.  
For all $i\in [1:S]$, the query is chosen as $Q_i^{(m)} \in \mathbb{F}_S^{M(S-T)}$.
We present an example for $n=1$.

In the query generation phase, 
a user generates a matrix $\mathcal{M}$ as follows:
\begin{align}
&\mathcal{M}=
\begin{pmatrix}
  \mathbf{r}[1]\quad\quad \dots \quad\quad  & \mathbf{r}[m] & \quad\quad \dots \quad\quad \mathbf{r}[M]\\
   \mathbf{O}_{(S-T)\times (m-1)(S-T)} & \mathbf{E}_{S-T} & \mathbf{O}_{(S-T) \times (M-m)(S-T)}
\end{pmatrix},
\end{align}
where $\mathbf{r}[1],\dots,\mathbf{r}[M]$ are $T \times (S-T)$ matrices.
For all $k \in [1:M]$, the factors of $\mathbf{r}[k]$ are chosen randomly from $\mathbb{F}_{S}$.
Next, for all $j \in [1:S]$, 
the user generates the query for $j$-th server as follows:
\begin{align}
Q_j^{(m)}=
\begin{pmatrix}
  1 & z_j & z_j^2 & \dots & z_j^{S-1} 
\end{pmatrix} 
\mathcal{M}. \label{wangquery}
\end{align}
$z_j$ are chosen from $\mathbb{F}_{S}$ to be different from each other.

In the response generation phase, 
for all $k\in [1:M]$, each server divides the messages into $S-T$ equal-length parts as follows:
\begin{equation}
  W_k=
(w_{k,1} ,\ \dots ,\ w_{k,S-T})^\top.
\end{equation}
The responses from $S$ servers are expressed as follows:
\begin{align} 
  &\begin{pmatrix}
    X_{1}^{(m)}\\
    X_{2}^{(m)}\\ 
    \vdots \\
    X_{S}^{(m)}    
  \end{pmatrix}
= 
\underbrace{
\begin{pmatrix}
  1 & z_1 & z_1^2 & \dots & z_1^{S-1}\\
  1 & z_2 & z_2^2 & \dots & z_2^{S-1}\\
  \vdots & \vdots & \vdots & \ddots & \vdots \\
  1 & z_S & z_S^2 & \dots & z_S^{S-1}\\    
\end{pmatrix} 
\mathcal{M}}_{Q_{1:S}^{(m)}}
\begin{pmatrix}
  W_{1}\\
  W_{2}\\
  \vdots\\
  W_{M}
\end{pmatrix}. \label{wangresponse}
\end{align}
Let $(\beta_{1},\ \dots,\ \beta_{T})^\top=\sum_{k\in[1:M]} \mathbf{r}[k] W_k$,
then the responses from the servers are transformed as follows:
\begin{align}
  \begin{pmatrix}
    X_{1}^{(m)}\\
    X_{2}^{(m)}\\ 
    \vdots \\
    X_{S}^{(m)}    
  \end{pmatrix}
= 
\begin{pmatrix}
  1 & z_1 & z_1^2 & \dots & z_1^{S-1}\\
  1 & z_2 & z_2^2 & \dots & z_2^{S-1}\\
  \vdots & \vdots & \vdots & \ddots & \vdots \\
  1 & z_S & z_S^2 & \dots & z_S^{S-1}\\    
\end{pmatrix} 
\begin{pmatrix}
  \beta_1\\
  \vdots\\
  \beta_T\\
  w_{m,1}\\
  \vdots\\
  w_{m,S-T}
\end{pmatrix}.\label{ss}
\end{align}
Since there are $T$ random numbers $\beta_1,\dots,\beta_T$, recovery is not possible with $T$ responses.
This forms the basis of the privacy mechanism.

In the decoding phase, 
since the Vandermonde matrix is a regular matrix,
the inverse matrix exists,
and the user decodes $W_m$ as follows:
\begin{align}
  \begin{pmatrix}  
    \beta_1\\
    \vdots\\
    \beta_T\\
    w_{m,1}\\
    \vdots \\ 
    w_{m,S-T}
  \end{pmatrix} = 
  \begin{pmatrix}
    1 & z_1 & z_1^2 & \dots & z_1^{S-1}\\
    1 & z_2 & z_2^2 & \dots & z_2^{S-1}\\
    \vdots & \vdots & \vdots & \ddots & \vdots \\
    1 & z_S & z_S^2 & \dots & z_S^{S-1}\\    
  \end{pmatrix} ^{-1}
  \begin{pmatrix}X_1^{(m)}\\X_2^{(m)}\\\vdots \\ X_S^{(m)}\end{pmatrix}.\label{wangdecode}
\end{align}
By obtaining the last $S-T$ symbols of (\ref{wangdecode}), the user retrieves $W_m = (w_{m,1}, \dots, w_{m,S-T})^\top$.

\subsection{Wang's Robust and Byzantine PIR \cite{8262858}}
In the scheme described in \ref{wan'scolludingPIR}, 
the order of the finite field is changed to a prime power $q \geq S+\mu$,
and the user uses $S+\mu$ servers, where $\mu=U$ for robust PIR and $\mu=2B$ for Byzantine PIR.
The responses of (\ref{wangresponse}) are changed to the following, and the responses can be regarded as MDS-encoded symbols:
\begin{align} 
  &\begin{pmatrix}
    X_{1}^{(m)}\\
    X_{2}^{(m)}\\ 
    \vdots \\
    X_{S+\mu}^{(m)}    
  \end{pmatrix}
= 
\underbrace{
\begin{pmatrix}
  1 & z_1 & z_1^2 & \dots & z_1^{S-1}\\
  1 & z_2 & z_2^2 & \dots & z_2^{S-1}\\
  \vdots & \vdots & \vdots & \ddots & \vdots \\
  1 & z_{S+\mu} & z_{S+\mu}^2 & \dots & z_{S+\mu}^{S-1}\\    
\end{pmatrix} 
\mathcal{M}}_{Q_{1:S+\mu}^{(m)}}
\begin{pmatrix}
  W_{1}\\
  W_{2}\\
  \vdots\\
  W_{M}
\end{pmatrix}. 
\end{align}
The user decodes $W_m$ by erasure correction for robust PIR and error correction for Byzantine PIR.
By setting $T=1$, this method is equivalent to non-colluding robust and Byzantine PIR.

\section{Conditions for PIR properties}\label{sec:conditions}
In this section, we derive the conditions for correctness and privacy of colluding, robust, and Byzantine PIR schemes.

 \subsection{Conditions for Correctness}
 For any $j \in [1:S]$, let the query to the $j$-th server be $Q_j^{(m)}=(Q_j^{(m)}[1],\dots,Q_j^{(m)}[M])$,
 the responses $X_1^{(m)},\dots,X_S^{(m)}$ are as follows:
 \begin{align}
 \begin{pmatrix}
 X_1^{(m)} \\ \vdots \\X_S^{(m)} 
 \end{pmatrix}
 =
 \begin{pmatrix}
   Q_1^{(m)}[1] & \dots & Q_1^{(m)}[M]  \\
  \vdots & \ddots & \vdots \\
   Q_S^{(m)}[1] & \dots & Q_S^{(m)}[M]
 \end{pmatrix}
 \begin{pmatrix}
   W_1 \\ \vdots \\ W_M 
 \end{pmatrix}.
 \end{align}
 Using a decoding matrix $D_m$, the decoding scheme is as follows:
 \begin{align}
 D_m
 \begin{pmatrix}
     X_1^{(m)} \\ \vdots \\X_S^{(m)} 
 \end{pmatrix}
 =
 W_m  .
 \end{align}
 Therefore, the condition for correctness is expressed as following theorem:
 \begin{theorem}[Condition for Correctness] \label{conditionsforcorrectness}
 there exists a matrix $D_m$ and the following equation should be satisfied.
 \begin{align}
    D_m \begin{pmatrix}Q_1^{(m)}\\\vdots \\ Q_S^{(m)}\end{pmatrix} = 
    \begin{pmatrix}
      \mathbf{O}_{(m-1)L_w \times L_w}\\
      \mathbf{E}\\
      \mathbf{O}_{(M-m)L_w \times L_w}
     \end{pmatrix}^{\top} \label{dec}.
 \end{align}
 The $\mathbf{O}_{l_1 \times l_2}$ denotes zero matrix of size $l_1 \times l_2$ and 
 the $\mathbf{E}$ denotes identity matrix of size $L_w \times L_w$.
\end{theorem}

\subsection{Conditions for Robust and Byzantine Privacy}
Let $F$ be a random factor of the query which is uniformly distributed over a certain finite set $\mathcal{F}$.
Assume that a query set $Q_{1:S}^{(m)}$ is constructed by a bijective function $\psi(m,F)$.
Furthermore, assume that the entropy of the entire query is constant regardless of the desired index $m$.

For all $j \in [1:S+U]$ in robust PIR or $j \in [1:S+2B]$ in Byzantine PIR, let $q_j$ be the query realization that the $j$-th server receives,
then the following lemma holds:
\newtheorem{lemma}{Lemma}
\begin{lemma}[Sufficient Condition for Robust and Byzantine Privacy] \label{sufficientlemma}
  For any $m,m^{\prime} \in [1:M]$, 
  $q_j \in \mathbf{Q}_j$, 
  the sufficient condition for \eqref{robustprivacy} and \eqref{byzantineprivacy} is as follows:
  \begin{align}
     \mathrm{Pr}\{\mathcal{Q}_j^{(m)}=q_j\}=\mathrm{Pr}\{\mathcal{Q}_j^{(m^{\prime})}=q_j\}\label{sufficient1}.
 \end{align}
\end{lemma}
\begin{proof}  
  Equation \eqref{sufficient1} indicates $\mathcal{Q}_j^{(\theta)}$ and $\theta$ are independent of each other,
  Therefore, from the assumption that the user selects $\theta$ independently of $\mathcal{W}_{1:S}$, 
  the following equation holds, and \eqref{robustprivacy} and \eqref{byzantineprivacy} are satisfied:
  \begin{align}
    I(\theta;\mathcal{Q}_j^{(\theta)}, \mathcal{W}_{1:S}) = 0.
  \end{align}
\end{proof}
\begin{lemma} [Uniqueness of the Desired Index]
  \label{uniq}
The following equation holds:
\begin{align}
   H(\theta | \mathcal{Q}_{1:S}^{(\theta)}) = 0.
\end{align}
\end{lemma}
\begin{proof}
  Let $p(q_{1:S},m,f)=\mathrm{Pr}\{\mathcal{Q}_{1:S}^{(\theta)}=q_{1:S}, \theta=m, F=f\}$ be a probability mass function of the queries,
  then the following equation holds:
  \begin{align}
  p(q_{1:S}|m, f) = 
  \begin{cases}
    1 \quad \mathrm{if}\ q_{1:S}=\psi(m,f)\\
    0 \quad \mathrm{otherwise} \ .
  \end{cases}\label{lem1}
    \end{align}
 Therefore, the following equations hold by using (\ref{lem1}):
\begin{align}
  & H(\theta | \mathcal{Q}_{1:S}^{(\theta)}) 
  \\ & = -\sum_{q_{1:S} \in \mathcal{Q}} \sum_{m \in [1:M]}  \sum_{f\in \mathcal{F}}p(q_{1:S},m,f)  \log p(m | q_{1:S}, f) \notag
   \\&= 0.\notag
\end{align}
\end{proof}

Furthermore, due to the assumptions that $F$ is uniformly distributed over $\mathcal{F}$, 
and $\mathcal{Q}_{1:S}^{(\theta)}$ is also chosen uniformly from $\mathbf{Q}_{1:S}$, 
then the following lemma holds for all $j \in [1:S+U]$ in robust PIR or $j \in [1:S+2B]$ in Byzantine PIR:

\begin{lemma} [Entropy of the Query for the Unresponsive or Byzantine Servers]
  \label{asm}
    For any $m,m^{\prime} \in [1:M]$, 
  $q_j \in \mathbf{Q}_j$, the following equation holds:
  \begin{align}
  H(\mathcal{Q}^{(m)}_{1:S}|\mathcal{Q}^{(m)}_{j}=q_{j}) = H(\mathcal{Q}^{(m^{\prime})}_{1:S}|\mathcal{Q}^{(m^{\prime})}_{j}=q_{j}). 
  \end{align}
\end{lemma}
\begin{proof}
  \begin{align}
      &\mathrm{Pr}\{\mathcal{Q}_{j}^{(\theta)}=q_{j} | \theta=m\} = \mathrm{Pr}\{\mathcal{Q}_{j}^{(\theta)}=q_{j} | \theta=m^{\prime}\}
      \\& \Rightarrow  H(\mathcal{Q}_{j}^{(\theta)} | \theta=m) = H(\mathcal{Q}_{j}^{(\theta)} | \theta=m^{\prime})
      \\& \Leftrightarrow  H(\mathcal{Q}^{(\theta)}_{1:S}|\mathcal{Q}^{(\theta)}_{j}=q_{j}, \theta=m)
       = H(\mathcal{Q}^{(\theta)}_{1:S}|\mathcal{Q}^{(\theta)}_{j}=q_{j}, \theta=m^{\prime}) \label{uni}
      \\& \Leftrightarrow  H(\mathcal{Q}^{(m)}_{1:S}|\mathcal{Q}^{(m)}_{j}=q_{j}) = H(\mathcal{Q}^{(m^{\prime})}_{1:S}|\mathcal{Q}^{(m^{\prime})}_{j}=q_{j}). 
    \end{align}    
   Equation (\ref{uni}) is due to 
   the assumption that the entropy of the entire query is constant regardless of the desired index.
   \end{proof}  
Let $\bar{j} = [1:S+U] \setminus \{j\}$ for robust PIR and $\bar{j} = [1:S+2B] \setminus \{j\}$ for Byzantine PIR, 
then the following theorem holds by using Lemma \ref{uniq} and Lemma \ref{asm}:
\begin{theorem} [Conditions for Robust and Byzantine Privacy] \label{privtheorem}
  Let $f\in \mathcal{F}$ be a realization of the random factor $F$.
  For any $m,m^{\prime} \in [1:M]$,
  $q_j \in \mathbf{Q}_j$, the following equations hold:
  \begin{align}
    &\mathrm{Pr}\{\mathcal{Q}_j^{(m)}=q_j\}=\mathrm{Pr}\{\mathcal{Q}_j^{(m^{\prime})}=q_j\} \label{probpriv}
    \\&\Leftrightarrow |\{f | \mathcal{Q}_j^{(m)}=q_j\}|=|\{f |\mathcal{Q}_j^{(m^\prime)}=q_j\}|\label{theoremrobust}. 
  \end{align}
  \end{theorem}
  \begin{proof}
    \begin{align}
    &\mathrm{Pr}\{\mathcal{Q}_j^{(m)}=q_j\}=\mathrm{Pr}\{\mathcal{Q}_j^{(m^{\prime})}=q_j\} 
    \\&\Leftrightarrow \mathrm{Pr}\{\mathcal{Q}_j^{(\theta)}=q_j | \theta=m \} = \mathrm{Pr}\{\mathcal{Q}_j^{(\theta)}=q_j  | \theta=m^\prime \}
           \\ &\Leftrightarrow \mathrm{Pr}\{\theta=m | \mathcal{Q}_j^{(\theta)}=q_j\} = \mathrm{Pr}\{\theta=m^{\prime} | \mathcal{Q}_j^{(\theta)}=q_j\}
    \\ &\Leftrightarrow \mathrm{Pr}\{\theta=m | \mathcal{Q}_j^{(\theta)}=q_j\} = \frac{1}{M} 
    \\ &\Leftrightarrow H(\theta | \mathcal{Q}_j^{(\theta)}=q_j) = \log M \label{p2}
    \\ &\Leftrightarrow H(\theta | \mathcal{Q}_j^{(\theta)}=q_j) - H(\theta | \mathcal{Q}_{\bar{j}}^{(\theta)},\mathcal{Q}_j^{(\theta)}=q_j)= \log M  \label{robustp3} 
    \\ &\Leftrightarrow I(\theta;\mathcal{Q}_{\bar{j}}^{(\theta)} | \mathcal{Q}_j^{(\theta)}=q_j) = \log M
       \\ &\Leftrightarrow H(\mathcal{Q}_{\bar{j}}^{(\theta)} | \mathcal{Q}_j^{(\theta)}=q_{j}) 
        =\log M + H(\mathcal{Q}_{\bar{j}}^{(m)} | \mathcal{Q}_j^{(m)}=q_{j}) \label{entrobust}
       \\& \Leftrightarrow |\{\mathcal{Q}_{1:S}^{(\theta)} | \mathcal{Q}_j^{(\theta)} = q_j\}|=M|\{\mathcal{Q}_{1:S}^{(m)} | \mathcal{Q}_j^{(m)} = q_j\}| \label{pri1}
       \\& \Leftrightarrow |\{\mathcal{Q}_{1:S}^{(m)} | \mathcal{Q}_j^{(m)}=q_{j}\}|=|\{\mathcal{Q}_{1:S}^{(m^\prime)} | \mathcal{Q}_j^{(m^\prime)}=q_{j}\}| 
       \\& \Leftrightarrow |\{\psi(m,f) | \mathcal{Q}_j^{(m)}=q_{j}\}| =|\{\psi(m^\prime,f) | \mathcal{Q}_j^{(m^\prime)}=q_{j}\}| 
       \\& \Leftrightarrow |\{f | \mathcal{Q}_j^{(m)}=q_{j}\}|=|\{f |\mathcal{Q}_j^{(m^\prime)}=q_{j}\}| . \label{conc1}
     \end{align}
     Equation (\ref{robustp3}) is due to Lemma \ref{uniq}, 
     and (\ref{entrobust}) is due to Lemma \ref{asm}.
     Equation (\ref{pri1}) is because (\ref{entrobust}) indicates the number of query realizations is $M$ times larger when $\theta$ is not given than when $\theta = m$,
     and also because $f$ is chosen uniformly from $\mathcal{F}$.
     Equation (\ref{conc1}) is because $\psi(\cdot)$ is a bijective function.
   \end{proof}
  This result indicates that the number of feasible $f$ must be constant for all $\theta$ even though a single query $q_j$ is obtained by the $j$-th server for any $j \in [1:S]$, so that the server can not infer the queries sent to the other servers.
 Due to Lemma \ref{sufficientlemma} and Theorem \ref{privtheorem}, (\ref{theoremrobust}) is the sufficient condition for privacy properties (\ref{robustprivacy}) and (\ref{byzantineprivacy}).
 However, some studies \cite{8720234,8423707,8849275} consider (\ref{probpriv}) as the privacy property.
 In that case, (\ref{theoremrobust}) is the necessary and sufficient condition for privacy.
 This condition means that the overall randomness of the queries estimated by the server to be equal, indicating that the privacy property can be evaluated based solely on the random elements of the queries.

 \subsection{Conditions for Robust and Byzantine Capacity}
We first derive the conditions for the capacity achievability of PIR, and then extend them to the conditions for robust and Byzantine PIR.
 Let $\pi$ be a random substitution for $[1:M]$.
  The capacity of PIR is derived using Lemma 5 and Lemma 6 in \cite{7889028}.
  The necessary and sufficient conditions for achieving the capacity are that the inequalities in the derivations of these two lemmas are satisfied with equal signs.
  Therefore, the following equations are the necessary and sufficient conditions for achieving the capacity:
  \begin{align}
    &H(\mathcal{X}_{1:S}^{(\pi[1])}|\mathcal{Q}_{1:S}^{(\pi[1])}) = \sum_{j\in[1:S]} H(\mathcal{X}_j^{(\pi[1])}|\mathcal{Q}_{1:S}^{(\pi[1])})\label{enum:m1},
    \\&I(\mathcal{W}_{\pi[m:M]} ; \mathcal{Q}_{1:S}^{(\pi[m-1])}, \mathcal{X}_{1:S}^{(\pi[m-1])}| \mathcal{W}_{\pi[1:m-1]}) 
    = \frac{1}{S} \sum_{j\in[1:S]} I(\mathcal{W}_{\pi[m:M]} ; \mathcal{Q}_{1:S}^{(\pi[m-1])}, \mathcal{X}_{j}^{(\pi[m-1])}| \mathcal{W}_{\pi[1:m-1]}) \label{enum:m2},
    \\&\frac{1}{S} \sum_{j\in [1:S]} H(\mathcal{X}_{j}^{(\pi[m])}| \mathcal{Q}_{1:S}^{(\pi[m])}, \mathcal{W}_{\pi[1:m-1]}) 
    = \frac{1}{S} \sum_{j\in [1:S]} H(\mathcal{X}_{j}^{(\pi[m])}| \mathcal{Q}_{1:S}^{(\pi[m])}, \mathcal{W}_{\pi[1:m-1]},\mathcal{X}_{1:j-1}^{(\pi[m])})\label{enum:m3}.
  \end{align}
  Let $\bar{m} = [1:M] \setminus \{m\}$,
  let $\bar{\mathcal{I}}$ be an arbitrary subset of $\bar{m}$, and let $\mathcal{I}$ be the complement of $\bar{\mathcal{I}}$.
  By using (\ref{enum:m1})-(\ref{enum:m3}), we can derive the following theorem.
\begin{lemma}[Conditions for Robust and Byzantine Capacity Achievability] \label{capacitylemma}
  The necessary and sufficient conditions for capacity-achieving PIR schemes are as follows:
  \begin{align}
  &\mathrm{rank}\ Q_{1:S}^{(m)}[\bar{m}]= \max_{j\in[1:S]}\ \mathrm{rank}\ Q_j^{(m)}[\bar{m}],\label{lemma4-1}\\
  &\mathrm{rank}\ Q_{1:S}^{(m)}[\mathcal{I}] = \sum_{j\in[1:S]} \mathrm{rank}\ Q_j^{(m)}[\mathcal{I}]\label{lemma4-2}.
  \end{align}
\end{lemma}
\begin{proof}
  Because the following equations hold, (\ref{enum:m3}) is a sufficient condition for (\ref{enum:m1}):
  
      \begin{align}
        &\frac{1}{S} \sum_{j\in [1:S]} H(\mathcal{X}_{j}^{(\pi[m])}| \mathcal{Q}_{1:S}^{(\pi[m])}, \mathcal{W}_{\pi[1:m-1]}) 
        = \frac{1}{S} \sum_{j\in [1:S]} H(\mathcal{X}_{j}^{(\pi[m])}| \mathcal{Q}_{1:S}^{(\pi[m])}, \mathcal{W}_{\pi[1:m-1]},\mathcal{X}_{1:j-1}^{(\pi[m])}) \\
        &\Rightarrow \sum_{j\in [1:S]} H(\mathcal{X}_{j}^{(\pi[1])}| \mathcal{Q}_{1:S}^{(\pi[1])}) 
        = \sum_{j\in [1:S]} H(\mathcal{X}_{j}^{(\pi[1])}| \mathcal{Q}_{1:S}^{(\pi[1])},\mathcal{X}_{1:j-1}^{(\pi[1])}) \label{c2}\\
        &\Leftrightarrow \sum_{j\in [1:S]} H(\mathcal{X}_{j}^{(\pi[1])}| \mathcal{Q}_{1:S}^{(\pi[1])}) = H(\mathcal{X}_{1:S}^{(\pi[1])}|\mathcal{Q}_{1:S}^{(\pi[1])}). \label{c3}
      \end{align}
  Therefore, (\ref{enum:m1}) is unnecessary.
  For (\ref{enum:m2}), the following equations hold for any $i,j\in [1:S]$:
      \begin{align}
        &I(\mathcal{W}_{\pi[m:M]} ; \mathcal{Q}_{1:S}^{(\pi[m-1])}, \mathcal{X}_{1:S}^{(\pi[m-1])}| \mathcal{W}_{\pi[1:m-1]}) 
        = \frac{1}{S} \sum_{j\in [1:S]} I(\mathcal{W}_{\pi[m:M]} ; \mathcal{Q}_{1:S}^{(\pi[m-1])}, \mathcal{X}_{j}^{(\pi[m-1])}| \mathcal{W}_{\pi[1:m-1]})
        \\&\Leftrightarrow H(\mathcal{X}_{1:S}^{(\pi[m-1])} | \mathcal{W}_{\pi[1:m-1]}, \mathcal{Q}_{1:S}^{(\pi[m-1])})
        = \frac{1}{S} \sum_{j\in [1:S]} H(\mathcal{X}_{j}^{(\pi[m-1])} | \mathcal{W}_{\pi[1:m-1]},\mathcal{Q}_{1:S}^{(\pi[m-1])}) 
        \\&\Leftrightarrow H(\mathcal{X}_j^{(\pi[m-1])} | \mathcal{W}_{\pi[1:m-1]},\mathcal{Q}_{1:S}^{(\pi[m-1])},\mathcal{X}_i^{(\pi[m-1])}) = 0 
        \\&\Leftrightarrow H(\mathcal{X}_j^{(\pi[m])} | \mathcal{W}_{\pi[1:m]},\mathcal{Q}_{1:S}^{(\pi[m])},\mathcal{X}_i^{(\pi[m])}) = 0 
        \\&\Leftrightarrow H(\mathcal{X}_j^{(m)} | \mathcal{W}_{m},\mathcal{Q}_{1:S}^{(m)},\mathcal{X}_i^{(m)}) = 0 .\label{m0}
    \end{align}
  We can summarize the conditions as follows:
  \begin{align}
    &H(\mathcal{X}_{j}^{(m)}|\mathcal{X}_i^{(m)},\mathcal{W}_{m}, \mathcal{Q}_{1:S}^{(m)}) = 0\label{enum:m5},
    \\&H(\mathcal{X}_{j}^{(m)}| \mathcal{X}_{i}^{(m)}, \mathcal{W}_{\bar{\mathcal{I}}},\mathcal{Q}_{1:S}^{(m)}) = H(\mathcal{X}_{j}^{(m)}|  \mathcal{W}_{\bar{\mathcal{I}}}, \mathcal{Q}_{1:S}^{(m)})\label{enum:m7}.
\end{align}
Notably, the query to the $j$-th server is expressed as $Q_j^{(m)}=(Q_j^{(m)}[1],\dots,Q_j^{(m)}[M])$.
  Let $\bar{m} = [1:M]\setminus \{m\}$, 
  For any $i,j\in [1:S]$, the server responses are given as follows:
  \begin{align}
   X_i^{(m)}=Q_i^{(m)} [\bar{m}]W_{\bar{m}} + Q_i^{(m)}[m]W_m, \\
   X_j^{(m)}=Q_j^{(m)} [\bar{m}]W_{\bar{m}} + Q_j^{(m)}[m]W_m.
  \end{align}
  Equation (\ref{enum:m5}) indicates $X_j^{(m)}$ is determined by $X_i^{(m)}$ 
  without $W_{\bar{m}}$, so $Q_j^{(m)}[\bar{m}]$ is a multiple of $Q_i^{(m)}[\bar{m}]$ or a zero matrix for any $j \in [1:S]$.  
  Therefore, (\ref{lemma4-1}) holds;
  the responses from the servers can be expressed as follows:
  \begin{align}
    X_i^{(m)}=Q_i^{(m)} [\bar{\mathcal{I}}]W_{\bar{\mathcal{I}}} + Q_i^{(m)}[\mathcal{I}]W_\mathcal{I}, \\
    X_j^{(m)}=Q_j^{(m)} [\bar{\mathcal{I}}]W_{\bar{\mathcal{I}}} + Q_j^{(m)}[\mathcal{I}]W_\mathcal{I}.
   \end{align}
   Equation (\ref{enum:m7}) indicates $X_j^{(m)}$ is not obtained by $X_i^{(m)}$ without knowing $W_\mathcal{I}$,
   so $Q_i^{(m)}[\mathcal{I}]$ and $Q_j^{(m)}[\mathcal{I}]$ are linearly independent or either is a zero matrix.
   Therefore, (\ref{lemma4-2}) holds.
\end{proof}
These conditions indicate that $Q_{j}^{(m)}[m]$ must vary from server to server, and $Q_{j}^{(m)}[\bar{m}]$ must be possible to cancel each other out for any $j \in [1:S]$, so that $W_m$ is recovered with fewer queries.

By using the results above, the conditions for capacity achievability of robust and Byzantine PIR are obtained as following theorem:
\begin{theorem} [Conditions for Robust and Byzantine Capacity] \label{conditionsforrobustcapacity} 
  The necessary and sufficient conditions for capacity-achieving robust and Byzantine PIR schemes are as follows:
  \begin{align}
    &\mathrm{rank}\ Q_{1:S+\mu}^{(m)}[\bar{m}]=\max_{j \in [1:S+\mu]}\ \mathrm{rank}\ Q_j^{(m)}[\bar{m}], 
  \end{align}
  \begin{align}
    &\forall \kappa \subset[1:S+\mu],\ |\kappa|=S, \notag
    \\&\mathrm{rank}\ Q_{\kappa}^{(m)}[\mathcal{I}]=\sum_{j \in \kappa} \mathrm{rank}\ Q_j^{(m)}[\mathcal{I}],
    \end{align}
    where $\mu=U$ for robust PIR, and $\mu=2B$ for Byzantine PIR.
\end{theorem} 
\begin{proof}
  The capacity of robust PIR is derived as follows \cite{8119895}:
  \begin{align}
  \mathcal{R} &= \frac{L_w}{\sum_{j\in [1:S+U]} H(\mathcal{X}_j^{(m)}|\mathcal{Q}_{1:S+U}^{(m)})} 
  \\&\leq  \frac{L_w}{\sum_{j\in \kappa, \kappa\subset [1:S+U],|\kappa|=S} H(\mathcal{X}_j^{(m)}|\mathcal{Q}_{1:S+U}^{(m)})} \label{rsing}
  \\&\leq  \frac{1-\frac{1}{S}}{1-(\frac{1}{S})^M}. \label{rcap}
\end{align}
Equation (\ref{rsing}) is because $\mathcal{X}_{1:S+U}$ is an erasure-correction code for recovering $U$ responses \cite{8119895}. 
The necessary and sufficient conditions for achieving the capacity are that these inequalities are satisfied with equal signs.
Equation (\ref{rsing}) is satisfied with an equal sign by using MDS encoding for erasure correction.
Equation (\ref{rcap}) is satisfied with an equal sign by satisfying the capacity of basic PIR with any $S$ servers.
Therefore, by using Lemma \ref{capacitylemma}, 
the conditions for the capacity of robust PIR are as follows:
\begin{align}
&\mathrm{rank}\ Q_{1:S+U}^{(m)}[\bar{m}]= \max_{j\in [1:S+U]}\mathrm{rank}\ Q_j^{(m)}[\bar{m}], 
  \end{align}
  \begin{align}
&\forall \kappa \subset[1:S+U],\ |\kappa|=S, \ 
\mathrm{rank}\ Q_{\kappa}^{(m)}[\mathcal{I}]=\sum_{j \in \kappa} \mathrm{rank}\ Q_j^{(m)}[\mathcal{I}].
\end{align}

The capacity of Byzantine PIR is derived as follows \cite{8457293}:
\begin{align}
  \mathcal{R} &= \frac{L_w}{\sum_{j\in [1:S+2B]} H(\mathcal{X}_j^{(m)}|\mathcal{Q}_{1:S+2B}^{(m)})} 
  \\&\leq \frac{S}{S+2B} \frac{L_w}{\sum_{j\in \kappa, \kappa\subset [1:S+2B],|\kappa|=S} H(\mathcal{X}_j^{(m)}|\mathcal{Q}_{1:S+2B}^{(m)})} \label{bsing}
  \\&\leq \frac{S}{S+2B} \frac{1-\frac{1}{S}}{1-(\frac{1}{S})^M}. \label{bcap}
\end{align}
  Equation (\ref{bsing}) is because $\mathcal{X}_{1:S+2B}$ is an error-correction code for correcting $B$ noisy responses \cite{8457293}.
The necessary and sufficient conditions for achieving the capacity are that these inequalities are satisfied with equal signs.
 Equation (\ref{bsing}) is satisfied with an equal sign by using MDS encoding for error correction.
 Equation (\ref{bcap}) is satisfied with an equal sign by satisfying the capacity of basic PIR in any $S$ servers.
Therefore, by using Lemma \ref{capacitylemma}, 
the conditions for the capacity of Byzantine PIR are as follows:
\begin{align}
  &\mathrm{rank}\ Q_{1:S+2B}^{(m)}[\bar{m}]=\max_{j\in [1:S+2B]}\ \mathrm{rank}\ Q_j^{(m)}[\bar{m}], 
  \end{align}
  \begin{align}
  &\forall \kappa \subset[1:S+2B],\ |\kappa|=S, \ 
  \mathrm{rank}\ Q_{\kappa}^{(m)}[\mathcal{I}]=\sum_{j \in \kappa} \mathrm{rank}\ Q_j^{(m)}[\mathcal{I}].
\end{align}
\end{proof}

 \subsection{Conditions for Colluding Privacy}
Let $\mathcal{T}$ be an arbitrary index set of $T$ colluding servers, 
then the following lemma holds:
\begin{lemma}[Sufficient Condition for Colluding Privacy] \label{colludingsufficientpriv}
  For any $m,m^{\prime} \in [1:M]$ and $q_{\mathcal{T}} \in \mathbf{Q}_{\mathcal{T}}$, 
  the sufficient condition for (\ref{privacy}) is obtained as follows:
  \begin{align}
     \mathrm{Pr}\{\mathcal{Q}_{\mathcal{T}}^{(m)}=q_\mathcal{T}\}=\mathrm{Pr}\{\mathcal{Q}_{\mathcal{T}}^{(m^{\prime})}=q_{\mathcal{T}}\}\label{colludingsufficient}.
 \end{align}
\end{lemma}
\begin{proof}  
 (\ref{colludingsufficient}) indicates $\mathcal{Q}_{\mathcal{T}}^{(\theta)}$ and $\theta$ are independent of each other,
  Therefore, because the user selects $\theta$ independently of $\mathcal{W}_{1:S}$, 
  and (\ref{privacy}) is satisfied.
\end{proof}
Furthermore, due to Lemma \ref{colludingsufficientpriv}, the following lemma holds:
\begin{lemma} [Entropy of the Query for the Colluding Servers]
  \label{asmcolluding}
  The following equation holds for any $m,m^{\prime} \in [1:M]$: 
  \begin{align}
  H(\mathcal{Q}^{(m)}_{1:S}|\mathcal{Q}^{(m)}_{\mathcal{T}}=q_{\mathcal{T}}) = H(\mathcal{Q}^{(m^{\prime})}_{1:S}|\mathcal{Q}^{(m^{\prime})}_{\mathcal{T}}=q_{\mathcal{T}}). 
  \end{align}
\end{lemma}
\begin{proof}
  \begin{align}
      &\mathrm{Pr}\{\mathcal{Q}_{\mathcal{T}}^{(\theta)}=q_{\mathcal{T}} | \theta=m\} = \mathrm{Pr}\{\mathcal{Q}_{\mathcal{T}}^{(\theta)}=q_{\mathcal{T}} | \theta=m^{\prime}\}
      \\& \Rightarrow  H(\mathcal{Q}_{\mathcal{T}}^{(\theta)} | \theta=m) = H(\mathcal{Q}_{\mathcal{T}}^{(\theta)} | \theta=m^{\prime})
      \\& \Leftrightarrow  H(\mathcal{Q}^{(\theta)}_{1:S}|\mathcal{Q}^{(\theta)}_{\mathcal{T}}=q_{\mathcal{T}}, \theta=m)
       = H(\mathcal{Q}^{(\theta)}_{1:S}|\mathcal{Q}^{(\theta)}_{\mathcal{T}}=q_{\mathcal{T}}, \theta=m^{\prime}) \label{unicolluding}
      \\& \Leftrightarrow  H(\mathcal{Q}^{(m)}_{1:S}|\mathcal{Q}^{(m)}_{\mathcal{T}}=q_{\mathcal{T}}) = H(\mathcal{Q}^{(m^{\prime})}_{1:S}|\mathcal{Q}^{(m^{\prime})}_{\mathcal{T}}=q_{\mathcal{T}}). 
    \end{align}    
   Equation (\ref{unicolluding}) is due to 
   the assumption that the entropy of the entire query is constant regardless of the desired index.
\end{proof}  

Let $\bar{\mathcal{T}} = [1:S]\setminus \mathcal{T}$, then the following theorem holds due to Lemma \ref{asmcolluding}:
\begin{theorem}[Conditions for Colluding Privacy] \label{conditionsforcolludingpriv} 
    Let $f\in \mathcal{F}$ be a realization of the random factor $F$.
    For any $m,m^{\prime} \in [1:M]$, 
  $q_{\mathcal{T}} \in \mathbf{Q}_{\mathcal{T}}$, the following equations hold:
  \begin{align}
    &\mathrm{Pr}\{\mathcal{Q}_\mathcal{T}^{(m)}=q_\mathcal{T}\}=\mathrm{Pr}\{\mathcal{Q}_\mathcal{T}^{(m^{\prime})}=q_\mathcal{T}\}
    \\&\Leftrightarrow |\{f | \mathcal{Q}_{\mathcal{T}}^{(m)} = q_{\mathcal{T}}\}|=|\{f | \mathcal{Q}_{\mathcal{T}}^{(m^\prime)} = q_{\mathcal{T}}\}|\label{theorem4}. 
  \end{align}
  \end{theorem}
  \begin{proof}
    \begin{align}
      &\mathrm{Pr}\{\mathcal{Q}_\mathcal{T}^{(m)}=q_\mathcal{T}\}=\mathrm{Pr}\{\mathcal{Q}_\mathcal{T}^{(m^{\prime})}=q_\mathcal{T}\}
       \\&\Leftrightarrow \mathrm{Pr}\{\mathcal{Q}_{\mathcal{T}}^{(\theta)}=q_{\mathcal{T}} | \theta=m\} = \mathrm{Pr}\{\mathcal{Q}_{\mathcal{T}}^{(\theta)}=q_{\mathcal{T}} | \theta=m^{\prime}\}
       \\ &\Leftrightarrow \mathrm{Pr}\{\theta=m | \mathcal{Q}_{\mathcal{T}}^{(\theta)}=q_{\mathcal{T}}\} = \mathrm{Pr}\{\theta=m^{\prime} | \mathcal{Q}_{\mathcal{T}}^{(\theta)}=q_{\mathcal{T}}\}
       \\ &\Leftrightarrow \mathrm{Pr}\{\theta=m | \mathcal{Q}_{\mathcal{T}}^{(\theta)}=q_{\mathcal{T}}\} = \frac{1}{M} \label{p1}
       \\ &\Leftrightarrow H(\theta | \mathcal{Q}_{\mathcal{T}}^{(\theta)}=q_{\mathcal{T}}) = \log M 
       \\ &\Leftrightarrow H(\theta | \mathcal{Q}_{\mathcal{T}}^{(\theta)}=q_{\mathcal{T}}) - H(\theta | \mathcal{Q}_{\bar{\mathcal{T}}}^{(\theta)},\mathcal{Q}_{\mathcal{T}}^{(\theta)}=q_{\mathcal{T}})= \log M  \label{p3}
       \\ &\Leftrightarrow I(\theta;\mathcal{Q}_{\bar{\mathcal{T}}}^{(\theta)} | \mathcal{Q}_{\mathcal{T}}^{(\theta)}=q_{\mathcal{T}}) = \log M
       \\ &\Leftrightarrow H(\mathcal{Q}_{\bar{\mathcal{T}}}^{(\theta)} | \mathcal{Q}_{\mathcal{T}}^{(\theta)}=q_{\mathcal{T}}) = \log M + H(\mathcal{Q}_{\bar{\mathcal{T}}}^{(m)} | \mathcal{Q}_{\mathcal{T}}^{(m)}=q_{\mathcal{T}}) \label{ent}
       \\& \Leftrightarrow |\{\mathcal{Q}_{1:S}^{(\theta)} | \mathcal{Q}_{\mathcal{T}}^{(\theta)} = q_{\mathcal{T}}\}|=M|\{\mathcal{Q}_{1:S}^{(m)} | \mathcal{Q}_{\mathcal{T}}^{(m)} = q_{\mathcal{T}}\}| \label{pri}
       \\& \Leftrightarrow |\{\mathcal{Q}_{1:S}^{(m)} | \mathcal{Q}_{\mathcal{T}}^{(m)} = q_{\mathcal{T}}\}|=|\{\mathcal{Q}_{1:S}^{(m^\prime)} | \mathcal{Q}_{\mathcal{T}}^{(m^\prime)} = q_{\mathcal{T}}\}| 
       \\& \Leftrightarrow |\{\psi(m,f) | \mathcal{Q}_{\mathcal{T}}^{(m)} = q_{\mathcal{T}}\}|=|\{\psi(m^\prime,f) | \mathcal{Q}_{\mathcal{T}}^{(m^\prime)} = q_{\mathcal{T}}\}| 
       \\& \Leftrightarrow |\{f | \mathcal{Q}_{\mathcal{T}}^{(m)} = q_{\mathcal{T}}\}|=|\{f | \mathcal{Q}_{\mathcal{T}}^{(m^\prime)} = q_{\mathcal{T}}\}| \label{key}.
     \end{align}
     Equation (\ref{p3}) is because of Lemma \ref{uniq}. Equation (\ref{ent}) is because of Lemma \ref{asmcolluding}. 
     Equation (\ref{pri}) is because (\ref{ent}) indicates the number of query realizations is $M$ times larger when $\theta$ is not given than when $\theta = m$. 
     Equation (\ref{key}) is because $\psi(m,f)$ is a bijective function.
   \end{proof}
   Due to Lemma \ref{colludingsufficientpriv} and Theorem \ref{conditionsforcolludingpriv}, (\ref{theorem4}) is the sufficient condition for privacy property (\ref{privacy}).
   This result indicates that the number of feasible $f$ must be constant for all $\theta$ even though $T$ queries are obtained by colluding servers, so that it can not infer queries sent to the other servers.

  \subsection{Conditions for Colluding Capacity}
  The capacity of colluding PIR is derived in \cite{8119895}, and the entire proof is given in \cite{sunarxiv}.
  Let $\pi$ be a random permutation over $[1:M]$,
  then the necessary and sufficient conditions for capacity achievability are satisfying seven inequalities in 
  \cite{sunarxiv} with equal signs as follows:
  \begin{align}
    &H(\mathcal{X}_{1:S}^{(\pi[m])}|\mathcal{W}_{\pi[1:m-1]},\mathbf{Q}_{1:S})
    =\frac{S}{\begin{pmatrix}S\\T\end{pmatrix}}\sum_{\mathcal{T}} \frac{H(\mathcal{X}_{\mathcal{T}}^{(\pi[m])} | \mathcal{W}_{\pi[1:m-1]},\mathbf{Q}_{1:S})}{T}.\label{han2}
\end{align}
  \begin{align}
    &H(\mathcal{X}_{\bar{\mathcal{T}}}^{(\pi[m])}|\mathcal{X}_{\mathcal{T}}^{(\pi[m-1])},\mathcal{W}_{\pi[1:m-1]},\mathbf{Q}_{1:S})
    =\sum_{n \in \mathcal{\bar{T}}} H(\mathcal{X}_n^{(\pi[m])} | \mathcal{X}_{\mathcal{T}}^{(\pi[m-1])},\mathcal{W}_{\pi[1:m-1]},\mathbf{Q}_{1:S}).
\end{align}
  \begin{align}
   &H(\mathcal{X}_{\mathcal{T}}^{(\pi[m])}|\mathcal{W}_{\pi[1:m+1]},\mathbf{Q}_{1:S})
   = \frac{1}{\begin{pmatrix}S\\T\end{pmatrix}} \sum_{\mathcal{T}} H(\mathcal{X}_{\mathcal{T}}^{(\pi[m+1])}|\mathcal{W}_{\pi[1:m]},\mathbf{Q}_{1:S}).
\end{align}
  \begin{align}
    &\forall n \notin \mathcal{T}, \forall \mathcal{T}^{\prime}, |\mathcal{T}^{\prime}|=T-1,\notag 
    \\& H(\mathcal{X}_n^{(\pi[m])} | \mathcal{X}_{\mathcal{T}}^{(\pi[m-1])},\mathcal{W}_{\pi[1:m-1]},\mathbf{Q}_{1:S})
    =H(\mathcal{X}_n^{(\pi[m])}|\mathcal{X}_{\mathcal{T}^{\prime}}^{(\pi[m-1])},\mathcal{W}_{\pi[1:m-1]},\mathbf{Q}_{1:S}).\label{ind5}
\end{align}
  \begin{align}
    &\forall n \in [1:S], \ 
    \sum_{n\in \mathcal{T}} H(\mathcal{X}_n^{(\pi[m])} | \mathcal{X}_{\mathcal{T}\setminus \{n\}}^{(\pi[m-1])},\mathcal{W}_{\pi[1:m-1]},\mathbf{Q}_{1:S})
    = H(\mathcal{X}_{\mathcal{T}}^{(\pi[m])}| \mathcal{W}_{\pi[1:m-1]},\mathbf{Q}_{1:S}).  \label{ind6}
\end{align}
  \begin{align}
    &H(\mathcal{X}_{\mathcal{T}}^{(\pi[m])} | \mathcal{W}_{\pi[1:m]},\mathbf{Q}_{1:S}) 
    = H(\mathcal{X}_{\mathcal{T}}^{(\pi[m])}, \mathcal{X}_{\bar{\mathcal{T}}}^{(\pi[m])}| \mathcal{W}_{\pi[1:m]},\mathbf{Q}_{1:S}). \label{dep} 
\end{align}
  \begin{align}
    &H(\mathcal{X}_{\bar{\mathcal{T}}}^{(\pi[\bar{m}])}|\mathcal{X}_{\mathcal{T}}^{(\pi[m])},\mathcal{X}_{\bar{\mathcal{T}}}^{(\pi[m])},\mathbf{Q}_{1:S}) 
    = H(\mathcal{X}_{\bar{\mathcal{T}}}^{(\pi[\bar{m}])}|\mathcal{X}_{\mathcal{T}}^{(\pi[m])},\mathcal{W}_{\pi[m]},\mathbf{Q}_{1:S}).\label{ind7}
\end{align}
Let $\bar{\mathcal{I}}$ be an arbitrary subset of $[1:M]\setminus \{m\}$, and let $\mathcal{I}$ be the complement of $\bar{\mathcal{I}}$,
then the following theorem holds:
\begin{theorem}[Conditions for Colluding Capacity Achievability] \label{conditionsforcolludingcapacity}
  From \eqref{han2}-\eqref{ind7},
  the necessary and sufficient conditions for capacity-achieving colluding PIR are expressed as follows:
  \begin{align}
    &\mathrm{rank}\ Q_{1:S}^{(m)}[\mathcal{I}]=\sum_{j=1}^{S} \mathrm{rank}\ Q_j^{(m)}[\mathcal{I}], \label{theorem2-1}
    \\&\mathrm{rank}\ Q_{1:S}^{(m)}[\bar{m}]=\max_{\mathcal{T}}\ \mathrm{rank}\ Q_{\mathcal{T}}^{(m)} [\bar{m}] . \label{theorem2-2}
    \end{align}
\end{theorem}
\begin{proof}
  When (\ref{ind5}) holds, the following equations also hold:
  \begin{align}
    &\forall i,j \in [1:S], i \neq j ,\ 
    I(\mathcal{X}_i^{(\pi[m])};\mathcal{X}_j^{(\pi[m])} | \mathcal{W}_{\pi[1:m-1]},\mathbf{Q}_{1:S})=0
    \\& \Leftrightarrow 
    \forall j \in [1:S], \forall m \in [1:M], \forall \bar{\mathcal{I}}\subset [1:M]\setminus\{m\}, \ 
    H(\mathcal{X}_{j}^{(m)}| \mathcal{Q}_{1:S}^{(m)}, \mathcal{W}_{\bar{\mathcal{I}}},\mathcal{X}_{i}^{(m)}) = H(\mathcal{X}_{j}^{(m)}| \mathcal{Q}_{1:S}^{(m)}, \mathcal{W}_{\bar{\mathcal{I}}}) \label{pi}. 
  \end{align}
  Due to (\ref{pi}), (\ref{han2})-(\ref{ind6}) are all satisfied. 
  Next, for (\ref{dep}), the following equations hold:
  \begin{align}
      &H(\mathcal{X}_{\mathcal{T}}^{(\pi[m])} | \mathcal{W}_{\pi[1:m]},\mathbf{Q}_{1:S})
      = H(\mathcal{X}_{\mathcal{T}}^{(\pi[m])}, \mathcal{X}_{\bar{\mathcal{T}}}^{(\pi[m])}| \mathcal{W}_{\pi[1:m]},\mathbf{Q}_{1:S}) 
      \\& \Leftrightarrow 
     H(\mathcal{X}_{\mathcal{T}}^{(m)} | \mathcal{W}_{m},\mathcal{Q}_{1:S}^{(m)}) 
      = H(\mathcal{X}_{\mathcal{T}}^{(m)}, \mathcal{X}_{\bar{\mathcal{T}}}^{(m)}| \mathcal{W}_{m},\mathcal{Q}_{1:S}^{(m)}) \label{pi2}.
  \end{align}
  Equation (\ref{pi2}) is because we can set $\pi[1]=m$.
  Therefore, (\ref{ind7}) is satisfied as follows:
  \begin{align}
     &H(\mathcal{X}_{\bar{\mathcal{T}}}^{(\pi[\bar{m}])}|\mathcal{X}_{\mathcal{T}}^{(\pi[m])},\mathcal{X}_{\bar{\mathcal{T}}}^{(\pi[m])},\mathbf{Q}_{1:S}) 
  \\ & =H(\mathcal{X}_{\bar{\mathcal{T}}}^{(\pi[\bar{m}])}|\mathcal{X}_{\mathcal{T}}^{(\pi[m])},\mathcal{X}_{\bar{\mathcal{T}}}^{(\pi[m])},\mathcal{W}_{\pi[m]} ,\mathbf{Q}_{1:S}) \label{indmes}
  \\ & =H(\mathcal{X}_{\bar{\mathcal{T}}}^{(\pi[\bar{m}])},\mathcal{X}_{\mathcal{T}}^{(\pi[m])},\mathcal{X}_{\bar{\mathcal{T}}}^{(\pi[m])}|\mathcal{W}_{\pi[m]} ,\mathbf{Q}_{1:S})
   -H(\mathcal{X}_{\mathcal{T}}^{(\pi[m])}, \mathcal{X}_{\bar{\mathcal{T}}}^{(\pi[m])}| \mathcal{W}_{\pi[m]},\mathcal{Q}_{1:S}^{(\pi[m])})
  \\ & =H(\mathcal{X}_{\bar{\mathcal{T}}}^{(\pi[\bar{m}])},\mathcal{X}_{\mathcal{T}}^{(\pi[m])}|\mathcal{W}_{\pi[m]} ,\mathbf{Q}_{1:S})
   -H(\mathcal{X}_{\mathcal{T}}^{(\pi[m])}| \mathcal{W}_{\pi[m]},\mathcal{Q}_{1:S}^{(\pi[m])})
  \\ & =H(\mathcal{X}_{\bar{\mathcal{T}}}^{(\pi[\bar{m}])} | \mathcal{X}_{\mathcal{T}}^{(\pi[m])}, \mathcal{W}_{\pi[m]} ,\mathbf{Q}_{1:S}) \label{sum2}.
  \end{align}
Thus, (\ref{pi}) and (\ref{pi2}) are the summary of (\ref{han2})-(\ref{ind7}).
Let $Q_j^{(m)}=(Q_j^{(m)}[1],\dots,Q_j^{(m)}[M])$ for all $j \in [1:S]$, then 
for any $i,j\in [1:S]$, the responses can be calculated as follows:
 \begin{align}
   X_i^{(m)}=Q_i^{(m)} [\bar{\mathcal{I}}]W_{\bar{\mathcal{I}}} + Q_i^{(m)}[\mathcal{I}]W_\mathcal{I}, \\
   X_j^{(m)}=Q_j^{(m)} [\bar{\mathcal{I}}]W_{\bar{\mathcal{I}}} + Q_j^{(m)}[\mathcal{I}]W_\mathcal{I}.
  \end{align}
  Equation (\ref{pi}) indicates $X_j^{(m)}$ is not obtained by $X_i^{(m)}$ without knowing $W_\mathcal{I}$,
  so the row vectors of $Q_i^{(m)}[\mathcal{I}]$ and $Q_j^{(m)}[\mathcal{I}]$ are linearly independent or either is a zero matrix,
and the following equation holds:
\begin{align}
 \mathrm{rank}\ Q_{1:S}^{(m)}[\mathcal{I}] =\sum_{j=1}^{S} \mathrm{rank}\ Q_j^{(m)}[\mathcal{I}].
\end{align}
   Furthermore, 
   for any $\bar{m} = [1:M]\setminus \{m\}$, the responses from the servers can also be calculated as follows:
   \begin{align}
    X_\mathcal{T}^{(m)}=Q_\mathcal{T}^{(m)} [\bar{m}]W_{\bar{m}} + Q_\mathcal{T}^{(m)}[m]W_m, \label{xt}\\
    X_{1:S}^{(m)}=Q_{1:S}^{(m)} [\bar{m}]W_{\bar{m}} + Q_{1:S}^{(m)}[m]W_m.\label{x1s}
   \end{align} 
Equation (\ref{pi2}) indicates $X_{1:S}^{(m)}$ is obtained by $X_\mathcal{T}^{(m)}$ without knowing $W_{\bar{m}}$,
so (\ref{xt}) and (\ref{x1s}) show that $Q_{1:S}^{(m)}[\bar{m}]$ must be determined by $Q_{\mathcal{T}}^{(m)}[\bar{m}]$ for any $\mathcal{T}$.
Therefore, the following equation holds: 
   \begin{align}
      \mathrm{rank}\ Q_{1:S}^{(m)}[\bar{m}] = \max_{\mathcal{T}}\ \mathrm{rank}\ Q_\mathcal{T}^{(m)} [\bar{m}].
   \end{align}
  \end{proof}
These conditions indicate that $Q_{1:S}^{(m)}[m]$ must vary from server to server, 
and $Q_{\mathcal{T}}^{(m)}[\bar{m}]$ must be possible to cancel other rows out for any $\mathcal{T}$, so that $W_m$ is recovered with fewer queries.

  \subsection{Summary}
The conditions for robust and Byzantine PIR properties are summarized in Table \ref{robustterms},
where $j \in [1:S+\mu], \kappa \subset[1:S+\mu],\ |\kappa|=S$, $\bar{\mathcal{I}} \subset [1:M]\setminus \{m\}$.
  \begin{table}[tb]
    \begin{center}
\caption{Conditions for Robust and Byzantine PIR} \label{robustterms}
\scalebox{1.05}[1.0]{
  \renewcommand{\arraystretch}{1.2}
\begin{tabular}{|c|c|}
  \hline
      Properties & Conditions \\ \Hline{1.0pt}
Correctness & 
 $\exists D_m,\ D_m \begin{pmatrix}Q_1^{(m)}\\\vdots \\ Q_S^{(m)}\end{pmatrix} = 
\begin{pmatrix}
  \mathbf{O}_{(m-1)L_w \times L_w}\\
  \mathbf{E}\\
  \mathbf{O}_{(M-m)L_w \times L_w}
 \end{pmatrix}^{\top}$
\\ \hline
Privacy & $|\{f | \mathcal{Q}_{j}^{(m)} = q_{j}\}|=|\{f | \mathcal{Q}_{j}^{(m^\prime)} = q_{j}\}|$\\ \hline 
Capacity &
$\mathrm{rank}\ Q_{\kappa}^{(m)}[\mathcal{I}]=\sum_{j \in \kappa} \mathrm{rank}\ Q_j^{(m)}[\mathcal{I}]$,\\
&$\mathrm{rank}\ Q_{1:S+\mu}^{(m)}[\bar{m}]=\max_{j\in [1:S+\mu]}\ \mathrm{rank}\ Q_{j}^{(m)} [\bar{m}]$\\\Hline{1.0pt}
\end{tabular}
}
\end{center}
\end{table}
  The conditions for colluding PIR properties are summarized in Table \ref{colludingterms},
where $\mathcal{T}$ satisfies $\mathcal{T} \subset [1:S], |\mathcal{T}|=T$.
  \begin{table}[tb]
    \begin{center}
\caption{Conditions for Colluding PIR} \label{colludingterms}
\scalebox{1.05}[1.0]{
  \renewcommand{\arraystretch}{1.2}
\begin{tabular}{|c|c|}
  \hline
      Properties & Conditions \\ \Hline{1.0pt}
Correctness & 
 $\exists D_m,\ D_m \begin{pmatrix}Q_1^{(m)}\\\vdots \\ Q_S^{(m)}\end{pmatrix} = 
\begin{pmatrix}
  \mathbf{O}_{(m-1)L_w \times L_w}\\
  \mathbf{E}\\
  \mathbf{O}_{(M-m)L_w \times L_w}
 \end{pmatrix}^{\top}$
\\ \hline
Privacy & $|\{f | \mathcal{Q}_{\mathcal{T}}^{(m)} = q_{\mathcal{T}}\}|=|\{f | \mathcal{Q}_{\mathcal{T}}^{(m^\prime)} = q_{\mathcal{T}}\}|$\\ \hline 
Capacity &
$\mathrm{rank}\ Q_{1:S}^{(m)}[\mathcal{I}]=\sum_{j=1}^{S} \mathrm{rank}\ Q_j^{(m)}[\mathcal{I}]$,\\ 
&$\mathrm{rank}\ Q_{1:S}^{(m)}[\bar{m}]=\max_{\mathcal{T}}\ \mathrm{rank}\ Q_{\mathcal{T}}^{(m)} [\bar{m}]$\\\Hline{1.0pt}
\end{tabular}
}
\end{center}
\end{table}

\section{Confirmation}
\label{sec:confirmation}
This section applies the derived conditions to the conventional adversarial PIR schemes described in Section \ref{sec:representative} to confirm whether they achieve capacity.
\subsection{Sun's Colluding and Robust PIR schemes}
Sun's colluding PIR scheme \cite{8119895} is a capacity-achieving scheme.
As shown in (\ref{sundecode}), the query is constructed so that the decoding matrix $D_m$ exists, thereby satisfying the correctness condition.
In this method, $F$ is induced by the randomness of $\mathbf{s}_{1},\dots,\mathbf{s}_{M}$.
Lemma 1 in \cite{8119895} shows that MDS encoding does not affect the distribution of the feasible query set of $\frac{T}{S}$ within the entire matrix.
Therefore, $T$ colluding servers cannot infer the queries sent to the other servers.
This means that the number of feasible $f$, which is the realization of $F$, is constant regardless of the desired index $m$. 
Thus, $|\{f | \mathcal{Q}_{\mathcal{T}}^{(m)} = q_{\mathcal{T}}\}|=|\{f | \mathcal{Q}_{\mathcal{T}}^{(m^\prime)} = q_{\mathcal{T}}\}|$ holds for any $\mathcal{T} \subset [1:S]$ satisfying $|\mathcal{T}| = T$, 
and the privacy condition is satisfied.
Furthermore, $\mathbf{e}_m$ in (\ref{sunsquery}) is a full-rank matrix.
Therefore, $\mathrm{rank}\ Q_{1:S}[\mathcal{I}] = \sum_{j\in[1:S]} \mathrm{rank}\ Q_{j}[\mathcal{I}]$ for any $j\in[1:S]$ and for any $\bar{\mathcal{I}} \subset [1:S]\setminus \{m\}$.
For any $k\in[1:M]\setminus \{m\}$, $\mathbf{e}_k$ are MDS-encoded by $\mathcal{G}(S^M, TS^{M-1})$. 
This means that $Q_{\mathcal{T}}[\bar{m}]$ can represent $Q_{1:S}[\bar{m}]$ for any $\mathcal{T} \subset [1:S]$ satisfying $|\mathcal{T}| = T$.
Therefore, $\mathrm{rank}\ Q_{1:S}[\bar{m}] = \max_{\mathcal{T}}\ \mathrm{rank}\ Q_{\mathcal{T}}[\bar{m}]$ holds, and the capacity conditions are satisfied.

Sun's robust PIR scheme \cite{8119895} is also a capacity-achieving scheme.
Any $S$ responses from the $S+U$ servers are sufficient to decode the desired message, since $\mathbf{e}_k$ are MDS encoded for all $k \in [1:M]$.
Therefore, the decoding matrix $D_m$ exists, and the correctness condition is satisfied.
Furthermore, due to Lemma 1 in \cite{8119895}, the number of feasible $f$ is constant regardless of the desired index $m$, and $|\{f | \mathcal{Q}_{j}^{(m)} = q_{j}\}|=|\{f | \mathcal{Q}_{j}^{(m^\prime)} = q_{j}\}|$ holds for any $j \in [1:S+\mu]$.
Thus, the privacy condition is satisfied.
Since $\mathcal{C}$ is an $((S+U)S^{M-1}, S^{M})$ MDS code and $\mathbf{s}_m$ is a full-rank matrix, any $S^{M}$ rows of $\mathbf{e}_m$ construct a full-rank matrix.
Therefore, $\mathrm{rank}\ Q_{1:S}[\mathcal{I}] = \sum_{j\in \kappa} \mathrm{rank}\ Q_{j}[\mathcal{I}]$ for any $\kappa \subset[1:S+\mu],\ |\kappa|=S$.
For any $k\in[1:M]\setminus \{m\}$, $\mathbf{e}_k$ are MDS-encoded by $\mathcal{G}((S+U)S^{M-1}, S^{M-1})$ in the non-colluding robust PIR. 
This means that $Q_{j}[\bar{m}]$ can represent $Q_{1:S+U}[\bar{m}]$ for any $j \in [1:S+U]$.
Therefore, $\mathrm{rank}\ Q_{1:S+U}[\bar{m}] = \max_{j\in[1:S+U]}\ \mathrm{rank}\ Q_{j}[\bar{m}]$ holds, and the capacity conditions are satisfied.

\subsection{Wang's Colluding and robust PIR}
Wang's colluding PIR scheme \cite{8262858} is not a capacity-achieving scheme.
As shown in (\ref{wangdecode}), the decoding matrix $D_m$ exists, and the correctness condition is satisfied.
In this method, $F$ is induced by the randomness of $\mathbf{r}[1],\dots,\mathbf{r}[M]$.
Any column vector of $\mathcal{M}$ in (\ref{wangquery}) has $T$ random elements,
and only $Q_{j}^{(m)}[m]=\begin{pmatrix}
  1 & z_j & z_j^2 & \dots & z_j^{S-1} 
\end{pmatrix} 
\mathcal{M}[m]$ has a non-random element in addition.
However, by collecting only $T$ queries, $Q_{\mathcal{T}}^{(m)}[m]$ has the same realizations as $Q_{\mathcal{T}}^{(m)}[k]$ for any $k \in [1:M]\setminus \{m\}$, because of the $T$ random elements.
Therefore, $|\{f | \mathcal{Q}_{\mathcal{T}}^{(m)} = Q_{\mathcal{T}}^{(m)}\}|=|\{f | \mathcal{Q}_{\mathcal{T}}^{(m^\prime)} = Q_{\mathcal{T}}^{(m)}\}|$ holds for any $\mathcal{T} \subset [1:S]$ satisfying $|\mathcal{T}|=T$, 
and the privacy condition is satisfied.
Since $Q_{1:S}^{(m)}[m]$ is an $S \times (S-T)$ matrix, $\mathrm{rank}\ Q_{1:S}^{(m)}[m] = S-T$. 
On the other hand, at least $S-1$ rows of $Q_{1:S}^{(m)}[m]$ are non-zero vectors.
Only when $z_{i}=0$ for one of $i\in[1:S]$, and the first row vector of $\mathcal{M}[m]$ is a zero vector, the $i$-th row vector of $Q_{1:S}^{(m)}[m]$ becomes a zero vector.
Therefore, $\sum_{j\in[1:S]} \mathrm{rank}\ Q_{j}^{(m)}[m] \geq S-1$.
Thus, $\mathrm{rank}\ Q_{1:S}^{(m)}[m] \neq \sum_{j\in[1:S]} \mathrm{rank}\ Q_{j}^{(m)}[m]$ and the conditions for capacity achievability are not satisfied.

As a fact, in this method, $S$ responses are downloaded to obtain $S-T$ segmented parts of the desired message.
However, if $z_{i}=0$ for one of $i\in[1:S]$, and the first row vector of $\mathcal{M}$ becomes a zero vector with probability $\frac{1}{S^{M(S-T)}}$, 
then 
a single query $Q^{(m)}_{i}$ also becomes a zero vector, 
and the user knows $X^{(m)}_{i}$ is $0$ without receiving it.
Therefore, $i$-th server does not need to return the response, and the download rate is obtained as follows: 
\begin{equation}
  \mathcal{R} = \frac{S-T}{S-1+(1-\frac{1}{S^{M(S-T)}})}=\frac{1-\frac{T}{S}}{1-\frac{1}{S^{M(S-T)+1}}}.
\end{equation}
This rate is lower than the colluding capacity.
In the case of robust PIR, the same discussion applies for any $S$ out of $S+U$ queries;
the capacity condition is not satisfied, while the correctness and privacy conditions are satisfied.

\section{Conclusion}
In this paper, we derived the necessary and sufficient conditions for privacy, correctness, 
and capacity achievability of colluding, robust, and Byzantine PIR. 
We represented the conditions for these PIR schemes in terms of query matrix properties.
Furthermore, we evaluated the conventional adversarial PIR schemes and confirmed whether they satisfy
the necessary and sufficient conditions for PIR.
The methods known as capacity-achieving PIR have been confirmed to indeed satisfy the conditions.
On the other hand, some methods that do not achieve capacity have been confirmed to not satisfy the conditions.

\bibliography{bib} 
\bibliographystyle{junsrt} 

\end{document}